\documentclass[a4paper,11pt]{article}

\usepackage{authblk}
\usepackage[utf8]{inputenc}
\usepackage{amsthm}

\usepackage{amsmath,amssymb,enumitem}
\usepackage[linesnumbered, ruled, vlined]{algorithm2e} 

\newtheorem{theorem}{Theorem}

\newtheorem{lemma}[theorem]{Lemma}
\newtheorem{corollary}[theorem]{Corollary}

\newtheorem{remark}{Remark}

\usepackage{microtype}
\usepackage{tikz}
\usetikzlibrary{arrows}
\usetikzlibrary{calc}

\usepackage{tikz, pgfplots}
\usetikzlibrary{patterns}
\usetikzlibrary{shapes.geometric, calc}

\usepackage[margin=1in]{geometry}
\usepackage{hyperref}
\usepackage{color}
\definecolor{darkgreen}{RGB}{0,100,0}
\definecolor{firebrick}{RGB}{178,34,34}

\newcommand{\E}{{\mathbb{E}}}

\newcommand{\pth}[1]{\ensuremath{\left(#1\right)}}

\newcommand{\e}{{\varepsilon}}

\usepackage{soul}

\newcommand{\betat}{\beta_t}
\newcommand{\gammat}{\gamma_t}
\newcommand{\cl}{{\rm cl}}
\newcommand{\abs}[1]{|{#1}|}

\newcommand{\F}{\mathcal{F}}
\newcommand{\A}{\mathcal{A}}
\newcommand{\n}{\mathcal{N}}
\newcommand{\ei}{e_i, \cdots , e_1}
\newcommand{\m}{e_{i+1}, \cdots, e_m}

\newcommand*\samethanks[1][\value{footnote}]{\footnotemark[#1]}

\title{Strong Collapse of Random Simplicial Complexes}
\date{}

\author[1]{Jean-Daniel Boissonnat
\thanks{
J.-D. Boissonnat and S. Pritam received funding from the
European Research Council (ERC) under the European Union’s Seventh Framework Programme (FP/2007-
2013) / ERC Grant Agreement No. 339025 GUDHI (Algorithmic Foundations of Geometry Understanding in
Higher Dimensions).}
}

\author[2]{Kunal Dutta
\thanks{K. Dutta and S. Dutta received funding from the Polish NCN SONATA Grant no. 2019/35/D/ST6/04525.}
}

\author[3]{Soumik Dutta
\samethanks[2]
}

\author[4]{Siddharth Pritam
\samethanks[1]
}

\affil[1]{Universit\'e C\^ote d'Azur, INRIA, Sophia Antipolis, France. email: Jean-Daniel.Boissonnat@inria.fr}
\affil[2]{Faculty of Mathematics, Informatics and Mechanics, University of Warsaw, Poland. email: K.dutta@mimuw.edu.pl}
\affil[3]{Faculty of Mathematics, Informatics, and Mechanics, University of Warsaw, Poland. email: s.dutta2@mimuw.edu.pl}
\affil[4]{Universit\'e C\^ote d'Azur, INRIA, Sophia Antipolis, France. email: siddharth.pritam@inria.fr}

\begin{document}

\maketitle

\begin{abstract}
    The \emph{strong collapse} of a simplicial complex, proposed by Barmak and Minian~\cite{StrongHomotopy}, is a combinatorial collapse of a complex onto its sub-complex. 
        Recently, it has received attention from computational topology researchers~\cite{StrongCollPerst,FlagCompStrongColl,FlagCompEdgeColl}, owing to its empirically observed usefulness in simplification and size-reduction
        of the size of simplicial complexes while preserving the homotopy class. We consider the strong collapse process on random simplicial complexes. For the 
        Erd\H{o}s-R\'enyi random clique complex $X(n,c/n)$ on $n$ vertices with edge probability $c/n$ with $c>1$, we show that after any maximal sequence of strong collapses
        the remaining subcomplex, or \emph{core} must have $(1-\gamma)(1-c\gamma) n+o(n)$ vertices asymptotically almost surely (a.a.s.), where $\gamma$ is the least non-negative fixed 
        point of the function $f(x) = \exp\pth{-c(1-x)}$ in the range $(0,1)$. 
        These are the first theoretical results proved for strong collapses on random (or non-random) simplicial complexes. \\

\end{abstract}

\section{Introduction}
\label{sec:intro}
	
        \subparagraph{Motivation}
Simple collapse is a combinatorial notion which simplifies a simplicial complex without changing its topology. It can be expressed as a series of elementary moves of removals of pair of simplices $\sigma$ and $\tau$, such that $\sigma$ is uniquely contained in $\tau$. 
The notion of simple collapse was introduced by J.H.C Whitehead ~\cite{Whitehead} to study homotopy types of cell complexes. Since then it has found usage in many different areas of topology, especially in computational topology. Recently new variants of simple collapses have been introduced, called strong collapses and more generally $d$-collapses~\cite{StrongHomotopy, FlagCompEdgeColl, AttaliLieutierExtendedCol}. In such collapses one removes special vertices (more generally $d$-simplices) called dominated vertices (simplices) whose link is a \textit{simplicial cone}.  It's again expressed as a series of elementary moves of removals of dominated vertices (simplices). They have been shown to be very powerful tools to solve many problems in computational topology. In particular, the recent works of Pritam et. al.~\cite{StrongCollapseSid, FlagCompEdgeColl, AngelKerberPritam2023} has shown that strong collapses and edge collapses ($d$-collapse for $d=1$) can be used for efficient computation of one parameter and multi-parameter persistence. Efficient computation of persistent homology is one of the central topic of research in topological data analysis. 

The computation of persistent homology involves computing homology groups of a nested sequence of simplicial complexes called \textit{filtrations}. And to compute persistent homology requires $\mathbb{O}(n^3)$ time and $\mathbb{O}(n^2)$ space, here $n$ is the total number of simplices in the filtration. The general technique developed in~\cite{StrongCollapseSid,FlagCompStrongColl, FlagCompEdgeColl, glisseSwapShiftTrim2022, AngelKerberPritam2023} is to reduce a filtration to a smaller filtration using strong or edge collapse such that the persistent homology is preserved. In~\cite{StrongCollapseSid,FlagCompStrongColl, FlagCompEdgeColl, glisseSwapShiftTrim2022, AngelKerberPritam2023}, it has been established through experiments that in practice the reduced filtrations are very small and thereafter computation of persistent homology is extremely fast. The gain in efficiency is quite dramatic in the case of \textit{flag (clique) complexes}, where the strong collapse and edge collapse can be computed using only the graph ($1$-skeleton) of the given complex \cite{FlagCompStrongColl, glisseSwapShiftTrim2022, AngelKerberPritam2023}. 

As mentioned above the efficiency reported in~\cite{StrongCollapseSid, FlagCompStrongColl, glisseSwapShiftTrim2022, AngelKerberPritam2023} are through experiments and there is no theoretical guarantee over the reduction size. This is due to the fact that in general the amount of reduction depends on the individual complex and its combinatorial structure. In fact, the reduction is dependent even on the order of the collapses and a different order can result in a different core, except in the case of strong collapse. Which is even harder when we want study the reduction size in a filtered simplicial complexes. This motivates us to consider the case of random simplicial complexes and study the average reduction size by collapses.

In this article, we study the problem of reduction size achieved by the strong collapses of a clique complex defined over an Erd\H{o}s-R\'enyi random graph. 

\paragraph*{Previous}
The study of random simplicial complexes was initiated in the seminal paper of Linial and Meshulam~\cite{DBLP:journals/combinatorica/LinialM06}. 
Later Meshulam and Wallach~\cite{DBLP:journals/rsa/MeshulamW09} generalized the model of random complexes to obtain the \emph{Linial-Meshulam} (LM) model of $d$-dimensional random complexes. 
Since then a large body of work from several authors has emerged 
on many different models of random simplicial complexes, studying various topological and geometric properties of such complexes~\cite{kahle2016random}.
The study of simple collapses for random simplicial complexes has also been of interest to researchers and there have been numerous works in this direction.
In the $d$-dimensional LM model, Kozlov~\cite{Kozlov} proved bounds on the threshold for vanishing of the $d$-th homology. Simple 
collapses on random complexes were first studied by Aronshtam, Linial, \L{}uczak and Meshulam~\cite{DBLP:journals/dcg/AronshtamLLM13}, who 
improved Kozlov's bound to get a tight bound on the threshold and also gave a bound on the threshold for collapsibility in the $d$-dimensional LM model. 
Later Aronshtam and Linial~\cite{DBLP:journals/rsa/AronshtamL15,DBLP:journals/rsa/AronshtamL16} extended this line of 
work, obtaining first the threshold for the vanishing of the $d$-th homology in~\cite{DBLP:journals/rsa/AronshtamL15} and then the threshold for 
non-collapsibility of the $d$-dimensional LM complex~\cite{DBLP:journals/rsa/AronshtamL16}. In \cite{linial2017random} Linial and Peled obtained precise asymptotic bounds 
on the size of the core of such complexes. Very recently, Malen~\cite{MALEN2023113267} has shown that the ER clique 
complex $X(n,p)$ is $(k+1)$-collapsible with high probability for $p=n^{-\alpha}$ when $\alpha> 1/(k+1)$. 

Thus, to the best of our knowledge, work on collapses in random complexes has so far considered only simple collapses.



\paragraph*{} Throughout this paper, we shall use the notation asymptotically almost surely (a.a.s.) for a series of events $(E_n)_{n\geq 1}$, when the probability of occurence of $E_n$ goes to 
$1$ as $n\to\infty$.
\paragraph*{Models of Random Simplicial Complexes}
In this paper we shall consider two models of random simplicial complexes, which are described below. 
For a graph $G$, let $\cl(G)$ denotes the clique (flag) complex on $G$, i.e. the simplicial complex where each complete subgraph of $G$ on $d$ vertices is a $d-1$-simplex in $\cl(G)$. 
The \emph{Erd\H{o}s-R\'enyi} (ER) model $X(n,p)$ on $n$ vertices with probability parameter $p\in [0,1]$ is 
given by connecting each possible pair of elements of an $n$-element set by an edge randomly and independently with probability $p$ to get the random graph $G=G(n,p)$. 
Let $X(n,p):=\cl(G(n,p))$.

\paragraph*{Our Contribution}
We give bounds on the size of the \textit{core} (i.e. the smallest sub complex without any dominated vertex) of a random simplicial complex after strong (vertex) collapses.
Whereas previous works focused on computing the threshold for the appearance and disappearance of the $k$-th homology class 
for different $k$, here we are more interested in computing the size of the core. We show that for $n$-vertex ER clique complexes, 
the size of the core after a maximal series of strong collapses is a.a.s. a constant fraction of $n$, with the 
constant depending only on the edge probability, and bounded away from $1$. Further, we also find a precise expression for the constant, 
as a fixed point of an implicit equation. Our first theorem is stated below.

Let $K$ be a simplicial complex. For $i \geq 1$, we use $f_{i}(K)$ to denote the set of $i-$simplices of the complex. By a slight abuse of notation, we let $f_0(K)$ denote the set of \emph{non-isolated} vertices of the 
complex and $\tilde{f_0}(K)$ denote the set of non-isolated vertices of the complex. (A vertex $v\in K$ is \emph{isolated} the if $v$ is not a face of any edge in $K$.) In a \emph{pruning phase} run on $K$, all dominated (strong-collapsible) vertices of $K$ are 
simultaneously collapsed. Let $R_{t}(K)$ denote a complex obtained by running $t$ pruning phases over $K$. Lastly, $R_{\infty}(K)$ denotes a complex obtained from $K$ 
after running a maximal series of strong collapses on $K$. i.e. the core of $K$.

\begin{theorem}
\label{sizeofcore}
Let $c>1$ and $X \sim X(n,c/n)$. Then there exists a constant $\gamma \equiv \gamma(c)$ given by the least non-negative fixed point of the function $f(x)=\exp\pth{-c(1-x)}$,  
$x\in (0,1)$, such that $0< \gamma < 1/c <1 $ and a.a.s. the following holds 
\[\abs{f_0(R_\infty(X))}=(1-\gamma)(1-c\gamma)n + o(n).\]
\end{theorem}

\begin{figure}[htb] 
                    \begin{center}
                     \includegraphics[width=100mm,height=70mm]{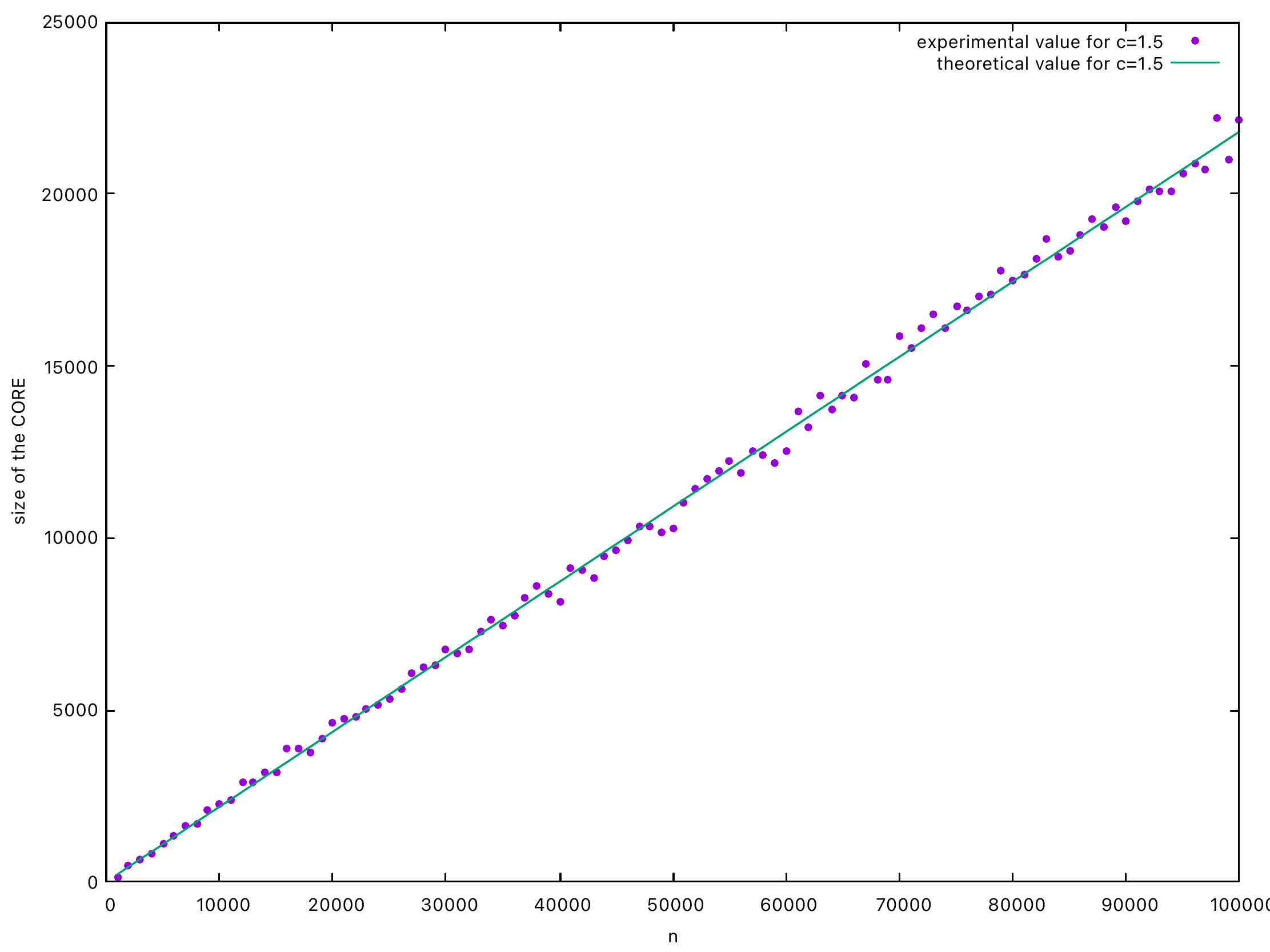}
			    \caption{We ran experiments on ER complexes and plotted the size of the core (strong collapse) with $n\in [0,10^5]$ and $c=1.5$. These experiments clearly validates our theoretical results.}
                            \label{fig:gr1-core-vs-n}
                     \end{center}
\end{figure}
 
\newpage 
 
\begin{figure}[htb] 
                    \begin{center}
                     \includegraphics[width=100mm,height=70mm]{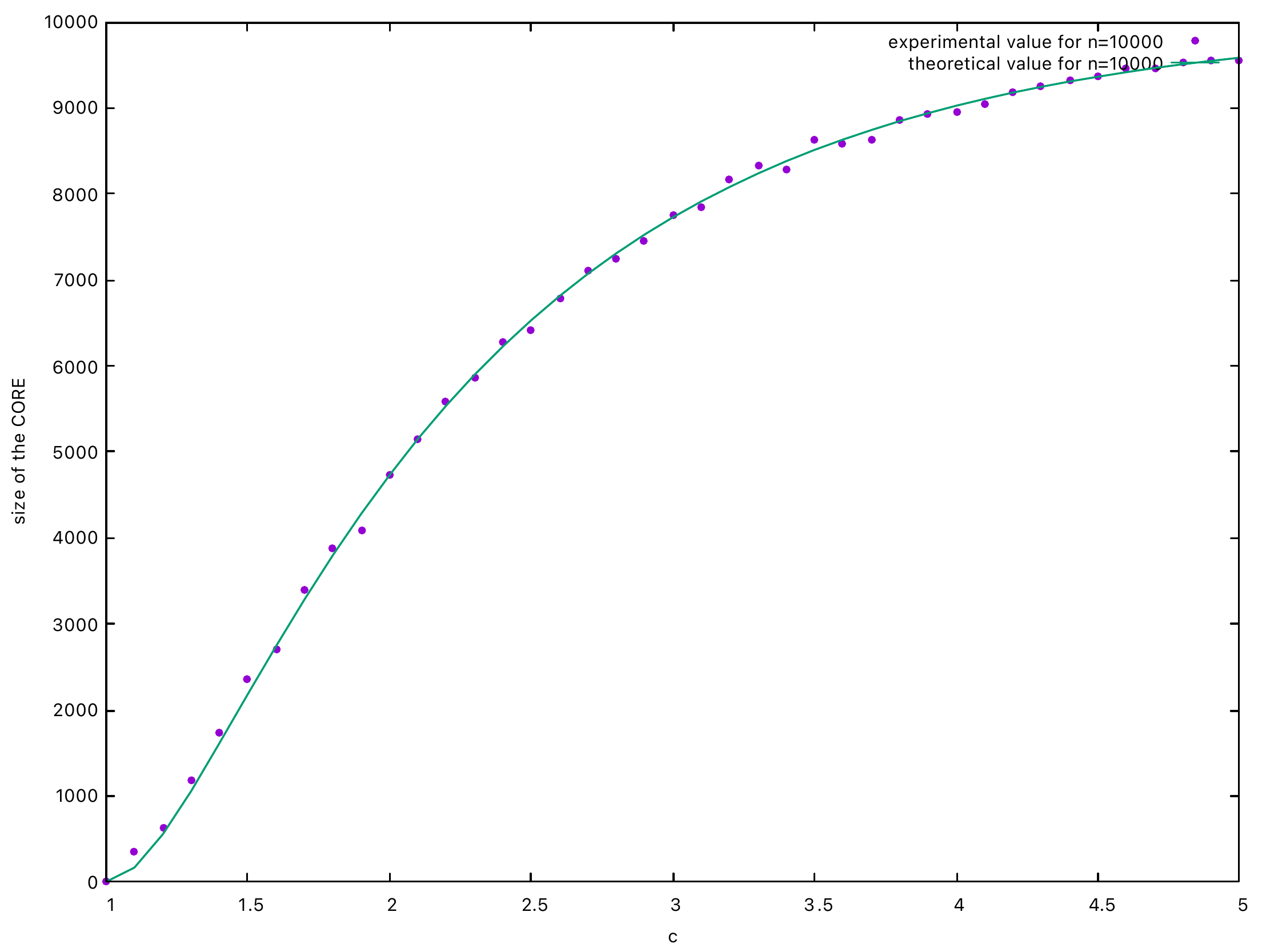}
			    \caption{In a different set of experiments over ER complexes, we varied the constant $c\in [1,5]$ keeping the number of vertices fixes to $n=10^{4}$.}
                            \label{fig:gr2-core-vs-c}
                     \end{center}
\end{figure}
 
We then address a question of algorithmic interest: given an $\e\in (0,1)$, how many rounds do we need in the first 
epoch to get within an $\e n$ gap from the actual size of the core? The following theorem gives a bound on the number of rounds $t$ as a function of $\e$.

\begin{theorem}
\label{thm:rate-conv-main}
Let $X \sim X(n,c/n)$ and $\e>0$ (sufficiently small) be given. Then there exists $t\in \mathbb{Z}_+$ such that 
\[\frac{(c-c\gamma)e^c}{1-ce^{-c}}(ce^{-c})^t \leq \epsilon \leq \frac{(c+1)e^c}{1-c\gamma}(c\gamma)^t,\]
and $\abs{f_0(R_t(X))}-\abs{f_0(R_\infty(X))} \leq \epsilon n + o(n)$ a.a.s. 
\end{theorem}

\subsection{Overview of Proofs and Outline of Sections}

While our analysis shares the general flow of the analysis of the simple collapsibility of LM complexes 
in e.g.~\cite{DBLP:journals/dcg/AronshtamLLM13,DBLP:journals/rsa/AronshtamL15,DBLP:journals/rsa/AronshtamL16}, 
there are several differences and difficulties. Firstly, in both cases (strong collapse of ER clique complex and simple collapse of LM model) our goal is to find the size 
of the core, rather than whether the complex is collapsible or not. Secondly, in the case of the random ER clique complex our analysis needs to take into account 
the non-homogeneity of the complex. That is, maximal simplices in this model can have different 
sizes. Further, unlike in the LM model, the existence of a maximal simplex is not independent of the existence of 
all other possible maximal simplices. Finally perhaps the most interesting difference of the random ER clique 
complex model in our context, is that the effect of removing a vertex is not necessarily localized -- a fact which 
requires a fair bit of innovation to handle (in several places), especially in proving the concentration bounds in Section~\ref{sec:concentration}, 
and in the later stages of the analysis, in Section~\ref{sec:second-epoch}. We present a more detailed overview of the proof strategy used in our concentration bound, in the 
beginning of Section~\ref{sec:concentration}.   \\ 

With the above caveats in mind, we first briefly review the main ideas of the proof of Aronshtam and Linial~\cite{DBLP:journals/rsa/AronshtamL16}. 
The analysis was split into two epochs, each of which were further divided into several rounds (phases). In the first epoch,
in each round, every simple-collapsible simplex was simultaneously collapsed, and this procedure was repeated for a constant
number of rounds. Aronshtam and Linial~\cite{DBLP:journals/rsa/AronshtamL16} used a tree-like model of a random simplicial complex to approximate the local structure of the 
random LM complex and showed that the total number of collapsed simplices over all such rounds tended to a constant fraction of the number of initial simplices, as the number of rounds increased. Moreover, this limit constant could be expressed as a fixed point 
of an implicit equation involving only the distribution parameter $c$. 

In the analysis of the second epoch, a simplex 
would be chosen randomly from the set of (non-neighbouring) simple-collapsible simplices, and collapsed in each round. 
The aim was to show that in this epoch, the number of simplices collapsed would be asymptotically negligible compared 
to the inital number of simplices present. Thus in summary, the final number of deleted simplices is determined by the 
first epoch itself, and the second epoch serves to show the tightness of this bound. \\
 
Our proofs also split the analysis into two epochs. Similar to~\cite{DBLP:journals/rsa/AronshtamL16},  
we show that a certain tree-like model of random simplicial complexes provides a good approximation of local neighbourhoods, which is done in Section~\ref{sec:tree}. 
This is followed by the analysis of the first epoch, in Section~\ref{sec:first-epoch}. The main theorem of this section gives an expression for 
the expected number of vertices remaining after $t$ rounds (or \emph{pruning phases}) of the first epoch.  
Bounds on $t$ as a function of $\e$, $\gamma$ and $c$, are given in Theorem~\ref{thm:rate-conv-main}, which is proved in Section~\ref{sec:rate-conv}.   
Before beginning the analysis of the second epoch however, we need bounds on the concentration of the size of the core itself, 
as well as several other random variables. These are proved in Section~\ref{sec:concentration}, where we use the notions of \emph{critical and 
precritical (sub)complexes} -- described in more detail in the beginning of Section~\ref{sec:concentration}. With these concentration bounds 
in place, we move to the analysis of the second epoch in Section~\ref{sec:second-epoch}.

\section{Preliminaries} \label{sect:prel}


In this section we briefly introduce some topological and probabilistic notions. Readers can refer to \cite{Hatcher} for a comprehensive introduction to topics related to topology 
and~\cite{Bollobas2001} for topics related to probability theory and random structures.

\subparagraph{Simplicial complex.} An \textbf{abstract simplicial complex} $\textit{K}$ is a collection of subsets of a non-empty finite set $\textit{X},$ such that for every subset $\textit{A}$ in $\textit{K}$, all the subsets of $\textit{A}$ are in $\textit{K}$. 
An element of $\textit{K}$ is called a \textbf{simplex}. An element of cardinality $d+1$ is called a $d$-simplex and 
$d$ is called its \textbf{dimension}. Given a simplicial complex $K$, we denote its geometric realization as $|K|$.
A simplex is called \textbf{maximal} if it is not a proper subset of any other simplex in $\textit{K}$. A sub-collection  $\textit{L}$ of $\textit{K}$ is called a \textbf{subcomplex} if it is a simplicial complex itself. A subcomplex $K'$ of $K$ is called a \textbf{ $d$-skeleton} of $K$ if it contains all the simplices of $K$ of dimension at most $d$.	


\subparagraph{Erdos Renyi Graph Definition. } This is the probability space $G(n,p)$ consisting of all the graphs on $n$ vertices. 
Probability of occurrence of a graph with $m$ edges is $p^m(1-p)^{(n)(n-1)/2-m}$. In other words, it is a random graph on $n$ vertices where each edge can occur independently with probability $p$.

\subparagraph{Clique complex and Neighborhood. } A complex $K$ is a \textbf{clique} or a \textbf{flag} complex if, when a subset of its vertices form a clique (i.e. any pair of vertices is joined by an edge), they span a simplex. 
For a vertex $v$ in $G$, the \textbf{open neighborhood} $N_G(v)$ of $v$ in $G$ is defined as  $N_G(v) := \{u \in G \: | \; [uv] \in E\}$, here $E$ is the set of edges of $G$. The \textbf{closed neighborhood} $N_G[v]$ is $N_G[v] := N_G(v) \cup \{ v\}$. 
Similarly we define the closed and open neighborhood of an edge $[xy] \in G$, $N_G[xy]$ and $N_G(xy)$ as $N_G[xy] := N[x] \cap N[y]$ and $N_G(xy) := N(x) \cap N(y)$, respectively. The above definitions can be extended to any $k$-clique 
$\sigma = [v_1, v_2, ..., v_k]$ of $G$; $N_G[\sigma] :=  \bigcap_{v_i \in \sigma}  N[v_i]$ and $N_G(\sigma) :=  \bigcap_{v_i \in \sigma}  N(v_i)$.

\subparagraph{Star, Link and Simplicial Cone. }
Let $\sigma$ be a simplex of a simplicial complex $K$, the \textbf{closed star} of $\sigma$ in $K$, $st_K(\sigma)$ is a subcomplex of $K$ which is defined as follows,
$ st_K(\sigma) := \{ \tau \in K | \hspace{5px} \tau \cup \sigma \in K \}.$
The \textbf{link} of $\sigma$ in $K$, $lk_K(\sigma)$ is defined as the set of simplices in $st_K(\sigma)$ which do not intersect with $\sigma$,
$ lk_K(\sigma) := \{ \tau \in st_K(\sigma) | \tau \cap \sigma = \emptyset \}.$
The \textbf{open star} of $\sigma$ in $K$, $st_K^o(\sigma)$ is defined as the set $st_K(\sigma) \setminus lk_K(\sigma)$. Usually $st_K^o(\sigma)$  is not a subcomplex of $K$. 

Let $L$ be a simplicial complex and let $a$ be a vertex not in $L$. Then  the set $aL$  defined as
$aL := \{ a,\hspace{2px} \tau \hspace{2px}|\hspace{5px} \tau \in L \hspace{5px} or \hspace{5px} \tau = \sigma \cup a; \hspace{5px} {\rm where} \hspace{5px} \sigma \in L \}$ 
is called a \textbf{simplicial cone}.

\subparagraph{Simple collapse.} Given a complex $K$, a simplex $\sigma \in K$ is called a \textbf{free simplex} if $\sigma$ has a unique coface $\tau \in K$. The pair $\{\sigma, \tau\}$ is called a \textbf{free pair}. The action of removing a free pair: $K\rightarrow K \setminus \{\sigma, \tau\}$ is called an \textbf{elementary simple collapse.} A series of such elementary simple collapses is called a \textbf{simple collapse}. We denote it as $K$ ${\searrow}$ $L$. 
A subcomplex $K^{ec}$ of $K$ is called an \textbf{elementary core} of $K$ if $K {\searrow} K^{ec}$ and $K^{ec}$ has no free pair. 

\subparagraph{Removal of a simplex.} We denote by \textbf{$K \setminus \sigma$} the subcomplex of $K$ obtained by removing $\sigma$, i.e. the complex that has all the simplices of $K$ except the simplex $\sigma$ and the cofaces of $\sigma$.  


\subparagraph{Dominated simplex.} 
A simplex $\sigma$ in $K$ is called a \textbf{dominated simplex} if the link  $lk_K(\sigma)$ of $\sigma$ in $K$ is a simplicial cone, i.e. if there exists a vertex $v^{\prime} \notin \sigma$ and a subcomplex $L$ of $K$, such that $lk_K(\sigma) = v^{\prime}L$. We say that the vertex  $v^{\prime}$ is \textit{dominating} $\sigma$ and that $\sigma$ is \textit{dominated} by $v^{\prime}$, which we denote as $\sigma \prec v^{\prime}$. 

\subparagraph{$\sigma$-algebra} The reader can refer to \cite{Bollobas2001} for the definition of $\sigma$-algebra.

\subparagraph{$d$-collapse.} 
Given a complex $K$, the action of removing a dominated $k$-simplex $\sigma$ from $K$ is called an \textbf{elementary $d$-collapse}, denoted as $K {\searrow\searrow}^d \{K \setminus \sigma\}$. A series of elementary $d$-collapses is called a \textbf{$d$-collapse}, denoted as $K$ ${\searrow\searrow}^k$ $L$. We further call a complex $K$ \textbf{$d$-collapse minimal} if it does not have any dominated $d$ simplices. A subcomplex $K^k$ of $K$ is called a \textbf{$d$-core} if $K$ ${\searrow\searrow}^k$ $K^d$ and $K^d$ is $d$-collapse minimal. A $0$-core of a complex $K$ is unique, however it is not true in general for $k \geq 1$. Like simple collapses, $d$-collapses preserve the homotopy type of a simplicial complex. 

A $0$-collapse is a \textbf{strong collapse} as introduced in \cite{StrongHomotopy} and $1$-collapse is called an \textbf{edge collapse}~\cite{FlagCompEdgeColl}. The following lemma from~\cite{FlagCompEdgeColl} characterizes the domination of a simplex in the special case of a flag complex in terms of neighborhood. 

\begin{lemma} \label{nbd_contenment}
Let  $\sigma$ be a simplex of a flag complex $K$. Then $\sigma$ will be dominated by a vertex $v^\prime$ if and only if $N_G[\sigma] \subseteq N_G[v^\prime]$. 
\end{lemma}

%

In this article, our main focus will be the case $d=0$, i.e. when $\sigma$ is a vertex. 
The next lemma from~\cite{FlagCompStrongColl}, though elementary, is of crucial significance. 
\begin{lemma} \label{flag_property}
Let $K$ be a flag complex and let $L$ be any subcomplex of $K$ obtained by strong collapses. Then $L$ is also a flag complex.
\end{lemma}

Both lemmas (Lemma~\ref{nbd_contenment} and Lemma~\ref{flag_property}) show that strong collapse is well-suited to flag complexes. In the next sections we will investigate the reduction capabilities of strong of a clique complex of an Erd\H{o}s-R\'enyi random graph.

\section{Tree process}\label{sec:tree}
In this section, we describe the tree process which is used to simulate the collapse process in the first epoch of the (strong) collapse.
A one-dimensional tree is built recursively as follows: 
\begin{enumerate} 
\item Start with a single node(root).
\item In the $n$th iteration, add children to all the leaves at distance $n-1$ from the root from Poisson distribution with parameter $c$($>1$).
\end{enumerate}
Let $\mathcal{T}_n$ denote the set of all possible trees after $n$th iteration for $n>1$ and $\mathcal{T}_0$ being the root itself. Let $\mathcal{T}:= \bigcup\limits_{n \in \mathbb{N}} \mathcal{T}_n $.

Let $\gammat$ be the probability that a tree $T \in \mathcal{T}_{t}$ is pruned to the root in \textit{no more than} $t-1$ steps. Clearly, \(\gamma_1=e^{-c}\). Set \(\gamma_0=0\). Also, we have the following recursive relation which is true in general: \[\gamma_{t+1}=e^{-c(1-\gammat)}.\]

Note that, in this process we never prune the root itself even if its degree is $1$. We call such a process \textit{root collapsing}. Let $\gamma_{t}^{k}$ denote the probability that a tree $T \in \mathcal{T}_t$ has degree $k$ after $t-1$ root collapsing steps. Then, \[\gamma_{t}^{k}=\frac{{(c(1-\gamma_{t-1}))}^k}{k!}e^{-c(1-\gamma_{t-1})}\]
Observe that $\gamma_{1}^{k}$ gives the initial degree distribution. Also, let $\gamma_{t}^{\geq 2}$ denote the probability that a vertex has degree atleast $2$ after $t-1$ root collapsing steps. Then we have  \[\gamma_{t}^{\geq 2}=\sum_{k=2}^{\infty} \gamma_{t}^{k} = 1-\gammat(1+c(1-\gamma_{t-1})).\]  

Define $\betat:=1-\gamma_{t+1}$. Thus, $\betat$ is the probability that \textit{atleast} $t+1$ root collapsing steps are needed to isolate the root of a tree $T \in \mathcal{T}_t$.  

Define $f(x):= e^{-c(1-x)}$ on the interval $[0,1]$. We shall assume $c>1$ for the rest of these paper unless specified otherwise. Note that $f([0,1]) \subseteq [0,1]$, $f^\prime(x)=cf(x)$ and $f(x)$ is  strictly increasing on the interval $[0,1]$. Let $f_t(x)$ denoted the function obtained by composing $f$ $t$ times. Then $f_t$ is also strictly increasing on $[0,1]$ for all $t\geq 1 $. As, $\gamma_1>\gamma_0$, applying $f_{t-1}$ on both sides, we get $\gammat > \gamma_{t-1}$ for all $t \geq 1$.

 Also define $\gamma$ to be the left most zero of the function $g(x):=e^{-c(1-x)}-x$ defined on the range $[0,\infty)$. Note that $0 < \gamma \leq 1$ as $g(1)=0$. 

\begin{lemma}
\label{drift1}
For $c>1$ and $\gamma=\gamma(c)$ defined as earlier we have $c\gamma<1$.
\end{lemma}
\begin{proof}[Proof of Lemma~\ref{drift1}]
Let $g(x)$ be defined as above. So $g(\gamma - )>0$ and $g(\gamma)=0$. Thus, from the differentiability of $g(x)$, $c\gamma -1= g^\prime(\gamma)\leq 0$. Now, if $c\gamma-1=0$ then $ce^{1-c}=1$, which is impossible for $c>1$. Thus $c\gamma -1 < 0$.
\end{proof}

Now observe that, $x \leq \gamma \implies f(x)\leq f(\gamma)=\gamma$. Thus, by the fact that $f^\prime(x)=cf(x)$, $f(x)$ restricted on $[0,\gamma]$ becomes a  contraction mapping. So, by Banach Fixed Point theorem, $f|_{[0,\gamma]}$ has an unique fixed point which, in our case, is $\gamma$.

 To summarize the above arguments we get the following remark.

\begin{remark}

$\gamma_t$ converges to $0<\gamma <1/c$ as an increasing sequence and $\betat$ converges to $1-1/c<\beta<1$ as a decreasing sequence. 
\end{remark}

\section{First Epoch}
\label{sec:first-epoch}

In this section, we present the analysis of the first epoch of the collapse. The first epoch is executed in phases and in each phase we remove a maximal set of dominated vertices simultaneously.

Our goal, in this section, is to prove the following theorem. 
\begin{theorem}
\label{expectation}
Let $X \sim X(n,c/n)$. Let $\E(\abs{f_0(R_t(X))})$ denote the expected number of non-isolated vertices in $X$ after $t$ strong collapse phases and $\gammat$ be as defined in the last section. Then,

\[\E(\abs{f_0(R_t(X))})=(1-\gamma_{t+1}-c\gammat+c\gammat^2)n.\]
\end{theorem}

We start by proving some important lemmas about the local structure of the complex. For \(X \sim X(n,\frac{c}{n})\), let us define the following event,  \[D:=\{deg(v)\leq \log n \,\,\,  \forall v \in \tilde{f_0}(X)\}.\] Then, the following lemma can be proved using standard Chernoff bounds.
\begin{lemma}
\label{l:prob-D-almost-1}
$Pr\{D\}=1-o_{n}(1)$.
\end{lemma}
\begin{proof}[Proof of Lemma~\ref{l:prob-D-almost-1}]
Note That for any $v \in \tilde{f_0}(X)$, $deg(v) \sim Bin(n-1,c/n)$. Let $D_v$ be the event that $deg(v)\leq \log(n)$. Then by the Chernoff Bound on Binomial Distribution

\begin{flalign*}
Pr\{\neg D_v\}
&=Pr\{deg(v) > \log(n)\}&\\
&\leq (\frac{ec(n-1)}{n \log(n)})^{\log(n)}&\\
&\leq (\frac{ec}{\log(n)})^{\log(n)}&\\
&=n^{-\log(\log(n/ec))}
\end{flalign*}
Thus by the union bound $Pr\{\bigcup_{v \in f_0(X)} \neg D_v\} \leq n^{1-\log(\log(n/ec))}$. Hence,

\[Pr\{D\}=Pr\{\bigcap_{v \in f_0(X)} D_v\}=1-Pr\{\bigcup_{v \in f_0(X)} \neg D_v\} \geq 1- n^{1-\log(\log(n/ec))}=1-o_n(1) \]

\end{proof}

By $Cl(S)$ we denote the simplicial closure of the set $S$.
Fix $v \in \tilde{f_0}(X)$. Define $\n_0:=Cl(v)$ and $\n_{-1}:=\emptyset$. Also define $\n_{i+1}:=Cl(\{s \in X|\exists u \in \tilde{f_0}(\n_i)\setminus \tilde{f_0}(\n_{i-1})| u \subset s\})\cup \n_i$. Equivalently, this can also be defined in terms of the $1$-skeleton of the complex.



Define the event $\A_t=\{\n_t \in \mathcal{T} \}$.

\begin{lemma}
\label{l:simp-compl-local-tree}
Let $ X \sim X(n,p)$ and fix $v \in \tilde{f_0}(X)$. Then $Pr\{\A_t \cap D\}=1-o(1)$  .
\end{lemma}
\begin{proof}
If $X \in D$ then $\tilde{f_0}(\n_{t})=O(log^{t+1}n)$ and 
we want to avoid $O(log^{2t+2}n)$ edges to make $\n_t$ a one dimensional tree. Probability of that happening is 
\[(1-c/n)^{O(log^{2t+2}n)}=1-o(1)\]
\end{proof}

Now the the degree of a node of this tree comes from $Bin(n-1,c/n)$. For large $n$ this distribution can be approximated by $Po(c)$.




\begin{proof}[Proof of Theorem~\ref{expectation}]
Recall that for a simplicial complex $X$, $R_{t}(X)$ denotes a complex obtained after $t$ 
phases and $f_{0}(X)$ denotes the set of non-isolated vertices ($0-$simplices) of the complex. 
Note that if a vertex $v \in f_0(X)$ survived $t$ pruning steps then it must have had degree $deg(v) \geq 2$ 
after $t-1$ pruning steps. Thus, $Pr\{v \in f_0(R_t(X))\} \leq \gamma_{t}^{\geq 2} = 1-\gammat(1+c(1-\gamma_{t-1}))$. 
This event counts both the isolated vertices and degree one, (i.e., collapsible) vertices. Thus this gives a slight over estimate. To get more precise estimate we observe that, in the spirit of \cite{DBLP:journals/rsa/AronshtamL16},  that a vertex $v \in f_0(X)$ survives $t$ pruning steps if it is neither collapsed nor isolated  after $t$ pruning steps. Probability of such an event is $1-\gamma_{t+1}+c\gammat-c\gammat^2$. The previous lemma asserts that it is indeed the survival probability of a vertex of the simplicial complex.
\end{proof}


\section{Rate of Convergence}
\label{sec:rate-conv}
In this section, we prove Theorem~\ref{thm:rate-conv-main} thus giving bounds on the rate of convergence of the variable $\gamma_t$.
\begin{lemma}
Let $c, \gammat$ be as defined earlier. Then \[c\gammat \leq \frac{\gamma_{t+1}-\gammat}{\gammat-\gamma_{t-1}}\leq c\gamma_{t+1}\]
\end{lemma}
\begin{proof}
Let $f(x)$ and $g(x)$ be as defined in section \ref{sec:tree}. Clearly,

\begin{flalign*}
\frac{\gamma_{t+1}-\gammat}{\gammat-\gamma_{t-1}}
&=\frac{g(\gammat)}{g(\gamma_{t-1})}\\
&=\frac{g(\gamma_{t-1}+(\gammat - \gamma_{t-1}))}{g(\gamma_{t-1})}\\
&=\frac{g(\gamma_{t-1})+g^\prime(s)(\gammat - \gamma_{t-1})}{g(\gamma_{t-1})} & \text{(for some $s \in [\gamma_{t-1},\gammat]$)}\\
&=f^\prime(s)=cf(s) & \text{(as $\gammat-\gamma_{t-1}=\gamma_{t-1}$)}
\end{flalign*}

As $f(x)$ is an increasing function the result follows.
\end{proof}

In particular, $ce^{-c} \leq \frac{\gamma_{t+1}-\gammat}{\gammat-\gamma_{t-1}}\leq c\gamma$ for all $t \geq 1$.

Let $\delta(t)= (1-\gamma_{t+1}-c\gammat+c\gammat^2) - (1-\gamma_{t+2}-c\gamma_{t+1}+c\gamma_{t+1}^2)$, as defined in section \ref{sec:second-epoch}. It can be shown that

\begin{flalign*}
\delta(t)
&= (1-\gamma_{t+1}-c\gammat+c\gammat^2) - (1-\gamma_{t+2}-c\gamma_{t+1}+c\gamma_{t+1}^2) \\
&= (\gamma_{t+2}-\gamma_{t+1})+ c(\gamma_{t+1}-\gammat)+c(\gammat+\gamma_{t+1})(\gamma_{t}-\gamma_{t+1})\\
&=(\frac{\gamma_{t+2}-\gamma_{t+1}}{\gamma_{t+1}-\gamma_{t}}+c-c(\gamma_{t+1}+\gamma_{t}))(\gamma_{t+1}-\gammat)
\end{flalign*}

Hence, \[(c-c\gamma)e^c(ce^{-c})^t \leq \delta(t) \leq (c+1)e^c(c\gamma)^t\]

Now define $\epsilon \equiv \epsilon(t):=(1-\gamma_{t+1}-c\gammat+c\gammat^2) - (1-\gamma-c\gamma+c\gamma^2)$ so that $\E[\abs{f_0(R_t(X))}]-(1-\gamma)(1-c\gamma]=\epsilon n$. Consequently, 

\begin{flalign*}
\epsilon(t)
&= (1-\gamma_{t+1}-c\gammat+c\gammat^2) - (1-\gamma-c\gamma+c\gamma^2) \\
&= (\gamma-\gamma_{t+1})+ c(\gamma-\gammat)+c(\gammat+\gamma)(\gamma_{t}-\gamma)\\
&=(\gamma-\gamma_{t+1})+(c-c(\gamma_{t}+\gamma))(\gamma-\gammat)
\end{flalign*}

So, \[\frac{(c-c\gamma)e^c}{1-ce^{-c}}(ce^{-c})^t \leq \epsilon(t) \leq \frac{(c+1)e^c}{1-c\gamma}(c\gamma)^t.\]

Thus we get the following corollary.

\begin{corollary}
for any $t\geq 1$ \[e^c(ce^{-c})^t \leq \gamma_{t+1} - \gammat \leq e^c(c\gamma)^t \]

and 

\[\frac{e^c}{1-ce^{-c}}(ce^{-c})^t \leq \gamma - \gammat \leq \frac{e^c}{1-c\gamma}(c\gamma)^t. \]
\end{corollary}
From the above corollary, we get the  Theorem~\ref{thm:rate-conv-main}.

\section{Concentration of Size of the Complex after the First Epoch}
\label{sec:concentration}
In this section, we shall prove a concentration bound on the size of the core. Unlike in the case of simple collapses in $d$-dimensional LM 
complexes~\cite{DBLP:journals/rsa/AronshtamL15,DBLP:journals/rsa/AronshtamL16}, concentration bounds in our case are less straightforward. 
Observe firstly, that deleting a single vertex $v$ could potentially change the domination status of an arbitrary number of vertices, as for example
when $v$ dominates the entire complex. Thus the influence of a vertex can be $n$ in the worst case. Therefore we shall need to use an edge exposure martingale inequality, 
in the form of a variant of an inequality of Freedman~\cite{10.1214/aop/1176996452}, given by Warnke~\cite{warnke_2016}, which allows us to consider the \emph{path variance} of the effect 
of a single edge, rather than the worst case effect.  \\

In order to bound the path variance, we shall show that if the influence of 
a variable is large, there is a specific class of subcomplexes, which we call \emph{Critical Complexes}, one of which must occur in the $1$-skeleton of the complex. 
It is not hard to show (and we do) that the probability of occurence of these subgraphs is vanishingly low in the original random complex. However, the variance needs to be controlled  
at all steps in the edge exposure martingale, i.e. when we are computing expectations over arbitrarily small subcomplexes of the original complex. To handle this, we need to define a superset 
of critical complexes, which we call \emph{Precritical Complexes}, and show that their probability of occurence will still be vanishingly small throughout the edge exposure process.
We can then define a stopped martingale which stops if at any step of the edge exposure process, a precritical complex occurs, and prove concentration bounds using Warnke's inequality
for this martingale. The final concentration bound is then the bound obtained for the stopped martingale, together with the probability that the martingale ever encounters a precritical complex.  \\

Fix $p=\frac{c}{n}$ and $m=\frac{n(n-1)}{2}$. For $1 \leq i \leq m$, $ e_i \sim Bernoulli(p)$ be i.i.d. random variables corresponding 
to existance of edges. Clearly $X(n,p) = e_1 \times \cdots \times e_m$ as probability spaces. Now we can define a filtration of 
$\sigma$-algebras $\{\F_i\}_{i=0}^{m}$ on $X(n,p)$ by setting $\F_i$ to the $\sigma$-algebra corresponding to $e_1, \cdots, e_i$.

Let $X \sim X(n,p)$. Now we construct an edge exposure martingale (see e.g.~\cite{AS-00} for a definition of the edge exposure martingale) as follows: 
Clearly, $Y_m=\abs{f_0(R_{t}(X))} $ and $Y_0=\E(\abs{f_0(R_{t}(X))} )$.

This section is devoted to prove the following concentration result, which says that the size of the complex after $t$ pruning rounds of the first epoch is close to 
its expected value with high probability. 

\begin{theorem}
\label{maintheorem}
(Main Theorem) Let $X \sim X(n,p)$. Let  $\abs{f_0(R_{t}(X))} $ be number of vertices after $t$ strong collapsing phases and 
$Y_0=\E(\abs{f_0(R_{t}(X))} )$ be its expected value. Then for any $s \geq 0$ we have,

\begin{eqnarray*}
  Pr\{\abs{\abs{f_0(R_{t}(X))}-Y_0}\geq s \cdot n^{\frac{2}{3}}\}&\leq& \\ 
      2\exp\pth{-\frac{s^2\cdot n^{\frac{1}{3}}}{(c4^{t+1}  + (2/3)sn^{-1/3}2^{t+1})+O(1/n)}}+O(1/n)&=&o_n(1).
\end{eqnarray*}
\end{theorem}

To prove this we begin by observing some combinatorial results. In the following 
lemmas, we show that the influence of deleting one vertex is bounded, with high probability.   
\begin{lemma}\label{geometriclemma1}
Pr\{deleting a vertex $b$ gives birth to $k$ \textit{newly generated} dominated vertices\} $\leq O(\frac{1}{n^{3k-4}}).$
\end{lemma}
\begin{proof}[Proof of Lemma \ref{geometriclemma1}]  We first claim that deleting a vertex $b$ gives birth to $k$ newly generated dominated vertices then $b$ atleast have $k$ neighbors one of which is the dominating vertex of $b$. In the following diagram, the solid arrow denotes domination and the white arrow denotes future domination in the next phase only after deleting vertex $b$. The pointy head of the arrow is towards the dominated vertex. Vertices $a,b,c$ may be connected to other vertices. The following diagrams exhibits some of the potential arrangements.

\begin{minipage}{.2\textwidth}
\begin{tikzpicture}[scale=0.4]

\tikzset{vertex/.style = {shape=circle,draw,minimum size=0.1 em}}
\tikzset{edge/.style = {-> ,> = triangle 60}}
\node[vertex] (a) at  (3,3) {$a$};
\node[vertex] (b) at  (3,0) {$b$};
\node[vertex] (c) at  (1,-2) {$c$};

\draw[edge] (a) to (b);

\tikzset{edge/.style = {-> ,> = open triangle 60 }}

\draw[edge] (a) to (c);

\tikzset{edge/.style = {- ,> = latex'}}

\draw[edge] (b) to (c);

\end{tikzpicture}

\end{minipage}
\begin{minipage}{.2\textwidth}
\begin{tikzpicture}[scale=0.4]

\tikzset{vertex/.style = {shape=circle,draw,minimum size=1.0 em}}
\tikzset{edge/.style = {-> ,> = triangle 60}}
\node[vertex] (a) at  (3,3) {$a$};
\node[vertex] (b) at  (3,0) {$b$};
\node[vertex] (c) at  (1,-2) {$c$};

\draw[edge] (a) to (b);

\tikzset{edge/.style = {-> ,> = open triangle 60 }}

\draw[edge] (c) to (a);

\tikzset{edge/.style = {- ,> = latex'}}

\draw[edge] (b) to (c);

\end{tikzpicture}

\end{minipage}
\begin{minipage}{.2\textwidth}

\begin{tikzpicture}[scale=0.4]

\tikzstyle{every node}=[font=\small]

\tikzset{vertex/.style = {shape=circle,draw,inner sep=0 pt, minimum size=15 pt}}
\tikzset{edge/.style = {-> ,> = triangle 60}}
\node[vertex] (a) at  (3,3) {$a$};
\node[vertex] (b) at  (3,0) {$b$};
\node[vertex] (c1) at  (1,-2) {$c_1$};
\node[vertex] (ci) at  (5,-2) {$c_i$};
\node[vertex] (cj) at  (8,-2) {$c_j$};

\draw[edge] (a) to (b);

\tikzset{edge/.style = {-> ,> = open triangle 60 }}

\draw[edge] (ci) to (cj);

\tikzset{edge/.style = {- ,> = latex'}}

\draw[edge] (b) to (c1);
\draw[edge] (a) to (c1);

\draw[edge] (b) to (ci);
\draw[edge] (a) to (ci);

\draw[edge] (b) to (cj);
\draw[edge] (a) to (cj);

\node at ($(c1)!.5!(ci)$) {\ldots};

\end{tikzpicture}
\end{minipage}

A careful inspection will show that these kind of arrangements are impossible. Indeed if it happens that will imply that the would-be-dominated vertices are already dominated. This is  because we are only deleting $b$ which is a common neighbor of all the would-be-dominated-dominating pairs. Thus neighbors of $b$ can not have white arrows between themselves.


Thus fig:1 gives the necessary minimal arrangements for the birth of $k$ newly generated dominated vertices. In the following diagram, all the $c_i$'s and their corresponding $d$'s are assumed to be connected to some non-neighbor of $a$ which lies in the set $\{e_1, \cdots, e_{l^\prime}\}$. We claim that in such case $f_1-f_0 \geq (k-2) + (k-1) + (k-1) \geq 3k-4$. We shall prove our claim by induction on $k\geq 2$. The case $k=2$ is evident from the following diagram.

\begin{center}
\begin{tikzpicture}[scale=0.4]

\tikzset{vertex/.style = {shape=circle,draw,minimum size=1.0 em}}
\tikzset{edge/.style = {-> ,> = triangle 60}}
\node[vertex] (a) at  (3,3) {$a$};
\node[vertex] (b) at  (3,0) {$b$};
\node[vertex] (c1) at  (1,-2) {$c_1$};
\node[vertex] (d) at  (3,-5) {$d$};
\node[vertex] (e) at  (7,-3) {$e$};

\draw[edge] (a) to (b);

\tikzset{edge/.style = {-> ,> = open triangle 60 }}

\draw[edge] (d) to (c1);

\tikzset{edge/.style = {- ,> = latex'}}

\draw[edge] (b) to (c1);
\draw[edge] (a) to (c1);

\draw[edge] (e) to (c1);
\draw[edge] (d) to (e);

\draw[edge] (a) to[bend right=90] (d);



\end{tikzpicture}
\end{center}

We now prove the induction step. Consider the following figure again. 

\begin{center} 
\begin{tikzpicture}[scale=0.4]

\tikzstyle{every node}=[font=\small]

\tikzset{vertex/.style = {shape=circle, inner sep=0 pt, draw, minimum size=15 pt}}
\tikzset{edge/.style = {-> ,> = triangle 60}}
\node[vertex] (a) at  (3,3) {$a$};
\node[vertex] (b) at  (3,0) {$b$};
\node[vertex] (c1) at  (1,-2) {$c_1$};
\node[vertex] (ck-2) at  (4,-2) {$c_{k-2}$};
\node[vertex] (ck-1) at (6,-2) {$c_{k-1}$};

\tikzset{vertex/.style = {draw,minimum size=10 pt}}
\node[vertex] (dset) at (4,-5) {$\{d_1, \cdots, d_l\}$};
\node[vertex] (eset) at (12,-4) {$\{e_1, \cdots, e_{l^\prime}\}$};
\draw[edge] (a) to (b);

\tikzset{edge/.style = {-> ,> = open triangle 60 }}

\draw[edge] (dset) to (c1);
\draw[edge] (dset) to (ck-2);
\draw[edge] (dset) to (ck-1);

\tikzset{edge/.style = {- ,> = latex'}}

\draw[edge] (b) to (c1);
\draw[edge] (a) to (c1);

\draw[edge] (b) to (ck-2);
\draw[edge] (a) to (ck-2);

\draw[edge] (b) to (ck-1);
\draw[edge] (a) to (ck-1);

\draw[edge] (eset) to (c1);
\draw[edge] (eset) to (ck-2);
\draw[edge] (eset) to (ck-1);
\draw[edge] (eset) to (dset);


\draw[edge] (a) to[bend right=100] (dset);
\draw[edge] (a) to[bend left=100] (dset);
\draw[edge] (a) to[bend left=100] (dset);


\node at ($(c1)!.5!(ck-2)$) {\ldots};

\end{tikzpicture}

\end{center}

Now assume that the claim holds for $k-1$. Now just adding the $kth$ vertex $c_{k-1}$ increases $f_1-f_0$ by $1$. Also the corresponding $d$ and $e$ increases  $f_1-f_0$ by $1$ each. This ends the induction step. So in the all the possible minimal arrangements $f_1-f_0 \geq 3k-4$. Thus expected number of such arrangements is $\binom{n}{f_0}\cdot(c/n)^{f_1}=O(\frac{1}{n^{3k-4}})$. Thus the result follows from Markov's inequality.
\end{proof}

\begin{corollary}
Pr\{deleting a vertex $b$ gives birth to $3$ \textit{newly generated} dominated vertices\} $\leq O(\frac{1}{n^{5}}).$
\end{corollary}

\begin{corollary}
Let $X \sim X(n,p)$ and $e \in f_{1}(X)$, then 

\[Pr \{ \abs{ f_{0}(R_{t}(X)) \setminus f_{0}(R_{t}(X \setminus \{e\}))  } \geq 2^{t+1}\} \leq O(\frac{1}{n^5}). \]

\end{corollary}
\begin{proof}
If such an event happens then there must be a dominated vertex in the process whose deletion creates atleast $3$ new dominating vertices. Thus the result follows from the previous corollary.
\end{proof}

Let \textit{Critical Complexes} denote the minimal simplicial complexes corresponding to $k=3$ (see the following diagrams). Let \textit{Precritical complexes}  be any of the Critical Complex without any four of the edges.. Let $N$ denote the set of complexes from $X(n,p)$ that contains a Critical Complex and $N^\prime$ denote the set of complexes contains a Precritical Complex. . Clearly, $N^\prime \supseteq N$. 

\begin{minipage}{.45\textwidth}
\begin{tikzpicture}[scale=0.4]

\tikzstyle{every node}=[font=\small]

\tikzset{vertex/.style = {shape=circle,draw,inner sep= 0 pt, minimum size=15 pt}}
\tikzset{edge/.style = {-> ,> = triangle 60}}
\node[vertex] (a) at  (3,3) {$a$};
\node[vertex] (b) at  (3,0) {$b$};
\node[vertex] (c1) at  (1,-2) {$c_1$};
\node[vertex] (c2) at  (5,-2) {$c_2$};
\node[vertex] (d) at  (3,-5) {$d$};
\node[vertex] (e) at  (3,-3) {$e$};

\draw[edge] (a) to (b);

\tikzset{edge/.style = {-> ,> = open triangle 60 }}

\draw[edge] (d) to (c1);
\draw[edge] (d) to (c2);

\tikzset{edge/.style = {- ,> = latex'}}

\draw[edge] (b) to (c1);
\draw[edge] (a) to (c1);

\draw[edge] (b) to (c2);
\draw[edge] (a) to (c2);

\draw[edge] (e) to (c1);
\draw[edge] (e) to (c2);
\draw[edge] (d) to (e);

\draw[edge] (a) to[bend right=90] (d);



\end{tikzpicture}

\end{minipage}
\begin{minipage}{.4\textwidth}
\begin{tikzpicture}[scale=0.4]

\tikzstyle{every node}=[font=\small]

\tikzset{vertex/.style = {shape=circle,draw,inner sep=0 pt, minimum size=15 pt}}
\tikzset{edge/.style = {-> ,> = triangle 60}}
\node[vertex] (a) at  (3,3) {$a$};
\node[vertex] (b) at  (3,0) {$b$};
\node[vertex] (c1) at  (1,-2) {$c_1$};
\node[vertex] (c2) at  (5,-2) {$c_2$};
\node[vertex] (d1) at  (1,-5) {$d_1$};
\node[vertex] (d2) at  (5,-5) {$d_2$};
\node[vertex] (e) at  (3,-3) {$e$};

\draw[edge] (a) to (b);

\tikzset{edge/.style = {-> ,> = open triangle 60 }}

\draw[edge] (d1) to (c1);
\draw[edge] (d2) to (c2);

\tikzset{edge/.style = {- ,> = latex'}}

\draw[edge] (b) to (c1);
\draw[edge] (a) to (c1);

\draw[edge] (b) to (c2);
\draw[edge] (a) to (c2);

\draw[edge] (e) to (d1);
\draw[edge] (e) to (d2);
\draw[edge] (e) to (c1);
\draw[edge] (e) to (c2);

\draw[edge] (a) to[bend right=50] (d1);
\draw[edge] (a) to[bend left=50] (d2);



\end{tikzpicture}
\end{minipage}


\begin{minipage}{.45\textwidth}
\begin{tikzpicture}[scale=0.4]

\tikzstyle{every node}=[font=\small]

\tikzset{vertex/.style = {shape=circle,draw,inner sep=0 pt, minimum size=15 pt}}
\tikzset{edge/.style = {-> ,> = triangle 60}}
\node[vertex] (a) at  (3,3) {$a$};
\node[vertex] (b) at  (3,0) {$b$};
\node[vertex] (c1) at  (0.5,-2) {$c_1$};
\node[vertex] (c2) at  (5.5,-2) {$c_2$};
\node[vertex] (d) at  (3,-5) {$d$};
\node[vertex] (e1) at  (2.2,-2.5) {$e_1$};
\node[vertex] (e2) at  (3.8,-2.5) {$e_2$};
\draw[edge] (a) to (b);

\tikzset{edge/.style = {-> ,> = open triangle 60 }}

\draw[edge] (d) to (c1);
\draw[edge] (d) to (c2);

\tikzset{edge/.style = {- ,> = latex'}}

\draw[edge] (b) to (c1);
\draw[edge] (a) to (c1);

\draw[edge] (b) to (c2);
\draw[edge] (a) to (c2);

\draw[edge] (e1) to (c1);
\draw[edge] (e2) to (c2);
\draw[edge] (d) to (e1);
\draw[edge] (d) to (e2);

\draw[edge] (a) to[bend right=90] (d);



\end{tikzpicture}

\end{minipage}
\begin{minipage}{.4\textwidth}
\begin{tikzpicture}[scale=0.4]

\tikzstyle{every node}=[font=\small]

\tikzset{vertex/.style = {shape=circle,draw,inner sep=0 pt, minimum size=15 pt}}
\tikzset{edge/.style = {-> ,> = triangle 60}}
\node[vertex] (a) at  (3,3) {$a$};
\node[vertex] (b) at  (3,0) {$b$};
\node[vertex] (c1) at  (1,-2) {$c_1$};
\node[vertex] (c2) at  (5,-2) {$c_2$};
\node[vertex] (d1) at  (1,-5) {$d_1$};
\node[vertex] (d2) at  (5,-5) {$d_2$};
\node[vertex] (e1) at  (2.2,-3) {$e_1$};
\node[vertex] (e2) at  (3.8,-3) {$e_2$};

\draw[edge] (a) to (b);

\tikzset{edge/.style = {-> ,> = open triangle 60 }}

\draw[edge] (d1) to (c1);
\draw[edge] (d2) to (c2);

\tikzset{edge/.style = {- ,> = latex'}}

\draw[edge] (b) to (c1);
\draw[edge] (a) to (c1);

\draw[edge] (b) to (c2);
\draw[edge] (a) to (c2);

\draw[edge] (e1) to (d1);
\draw[edge] (e2) to (d2);
\draw[edge] (e1) to (c1);
\draw[edge] (e2) to (c2);

\draw[edge] (a) to[bend right=50] (d1);
\draw[edge] (a) to[bend left=50] (d2);



\end{tikzpicture}
\end{minipage}

The following result is immediate. 

\begin{lemma}
\label{criticallemma}
Let $X \sim X(n,p)$. Then,
\[Pr \{ X \in N \} \leq O(\frac{1}{n^5}),\]
and
\[Pr \{ X \in N^\prime \} \leq O(\frac{1}{n}).\]
\end{lemma}

Now define stopping time $\tau$ on $\{\F_i\}_{i=0}^{m}$ such that $\tau=t$ if $t=\min_{(s \leq t)}\{\F_s \in N^\prime\}$. Define a stopped martingale with respect to $\{\F_i\}_{i=0}^{m}$ by $M_i := Y_{i \wedge \tau}$. 

We first prove the following theorem.

\begin{theorem}\label{SMI}(Stopped Martingale inequality) Let $ \{M_i\}_{i=0}^{m} $ be the stopped martingale defined as above. Then for any $s \geq 0$ we have,

 \[Pr\{\abs{M_m-M_0}\geq s\cdot n^{2/3}\}\leq 2\exp\pth{-\frac{s^2\cdot n^{1/3}}{(c4^{t+1}  + (2/3)sn^{-1/3}2^{t+1})+O(1/n)}}.\]

\end{theorem}

In order to prove the above theorem, we shall use the following lemma from Warnke~\cite{warnke_2016}. 
Assume that $\{\F_K\}_{0\leq k \leq N}$ is an increasing sequence of $\sigma$-algebras, and  $\{M_K\}_{0\leq k \leq N}$ is an $\{\F_K\}_{0\leq k \leq N}$-adapted bounded martingale.



\begin{lemma}
\label{BVMI}
(2-sided version of Bounded Variance martingale Inequality) Let $U_k$ be a $F_{k-1}$ variable satisfying $\abs{M_k-M_{k-1}}\leq U_k$. Set $C_k=\max_{i\in [k]}U_k$ and $V_k=\sum_{i\in [k]}V(M_i-M_{i-1}|F_{i-1})$. Let $\phi(x)=(1+x)\log(1+x)-x$. For every $s\geq 0$ and $V,C > 0$ we have \[ Pr\{\abs{M_K - M_0} \geq s, V_k\leq V,  C_k \leq C \: for \: some \: k \in [N] \}\leq 2e^{-s^2/(2V+2Cs/3)}\]
\end{lemma}

Theorem~\ref{maintheorem} essentially follows from Theorem \ref{SMI} and Lemma \ref{criticallemma}.

\begin{proof}\textit{Proof of Theorem \ref{SMI}}

Let $X \in X(n,p)$. We shall first try to calculate $Pr\{X \in N| e_i, \cdots , e_1\}$ where $(e_1, \cdots , e_i)$ does not form any precritical complex. Let $M \subset N$ be the set of complexes that contains some critical complex \textit{not} involving any of the edges from $\{\ei\}$ and $M^{\prime} \subset N$ be the set of complexes where all the critical complexes involves some edges from  $\{\ei\}$. Clearly, $N = M \sqcup M^{\prime}$. Thus,

\begin{flalign*}
Pr\{X \in N| \ei\}
&= Pr\{X \in M \sqcup M^{\prime}| \ei\} &\\
&= Pr\{ X \in M  |\ei \} + Pr\{ X \in M^{\prime}| \ei \}&\\
&= Pr\{ X \in M  \} + Pr\{  X \in M^{\prime}| \ei \}
\end{flalign*}

$Pr\{ X \in M\}$ is the probability that $\{ \m \}$ contains a critical complex. By reasoning similar to the proof of Lemma \ref{geometriclemma1}. We get $Pr\{X \in M\}=O(1/n^5)$. On the other hand, note that as $(\ei)$ does not contain any precrtitical complex, atleast $5$ more edges is needed for $X$ to form a critical complex involving  some edges $\{\ei \}$. Suppose, depending on  $(e_1, \cdots , e_i)$, $k$ more edges are needed to complete a critical complex. Clearly $5 \leq k \leq 13$. Also note that, in a critical complex, there are atmost $2$ vertices of degree two and rests have degree atleat $3$. Thus even in the worst case one need to choose $3$ vertices and construct $5$ particular edges.  Thus $Pr\{  X \in M^{\prime}| \ei \}=O(1/n^2)$. Hence, $Pr\{X \in N| e_i, \cdots , e_1\}=O(1/n^2)$ given $(e_1, \cdots , e_i)$ does not form any precritical complex. In particular, \[Pr \{ \abs{ f_{0}(R_{t}(X)) \setminus f_{0}(R_{t}(X \setminus \{e_{i+1}\}))  } \geq 2^{t+1}| \ei \} \leq O(\frac{1}{n^2}) \]  under the same assumption.

Thus  \[\E( \abs{ f_{0}(R_{t}(X)) \setminus f_{0}(R_{t}(X \setminus \{e_{i+1}\})) } | \ei    ) \leq 2^{t+1} + n\cdot O(1/n^2) \leq 2^{t+1} + O(1/n)  \] whenever $(\ei)$ does not contain any precrtitical complex.

We now claim that $\abs{M_{i+1}-M_i}\leq 2^{t+1}+ O(1/n)$. If $(\ei)$ contains a precritical complex then the martingale stops and the claim holds. Now suppose $(\ei)$ does not contain any precritical complex. Then 

\begin{flalign*}
\abs{[M_{i+1}-M_i](1, \ei)}
&=\abs{M_{i+1}(1, \ei )-\E_{e_{i+1}}[M_{i+1}]}&\\
&=\abs{M_{i+1}(1, \ei )-p\cdot(M_{i+1}(1, \ei))-(1-p)\cdot(M_{i+1}(0, \ei))}&\\
&=\abs{(1-p)\cdot(M_{i+1}(1, \ei )-M_{i+1}(0, \ei))}&\\
&\leq (1-p)(2^{t+1} + O(1/n))
\end{flalign*}

Similarly,

\begin{flalign*}
\abs{[M_{i+1}-M_i](0, \ei)}
&=\abs{M_{i+1}(0, \ei )-\E_{e_{i+1}}[M_{i+1}]}&\\
&=\abs{M_{i+1}(0, \ei )-p\cdot(M_{i+1}(1, \ei))-(1-p)\cdot(M_{i+1}(0, \ei))}&\\
&=\abs{p\cdot(M_{i+1}(0, \ei )-M_{i+1}(1, \ei))}&\\
& \leq p(2^{t+1} + O(1/n) )
\end{flalign*}

 Hence the claim follows.

Next we claim that $Var(M_{i+1}-M_i|\F_i)=Var(M_{i+1}|\F_i)\leq (c/n)(4^{t+1} + O(1/n)) \leq O(1/n)$. Indeed if $(\ei)$ contains a precritical complex then the martingale stops and the variance is zero. Otherwise

\begin{flalign*}
Var(M_{i+1}-M_i|\ei)
&= p\cdot ([M_{i+1}-M_i](1, \ei))^2 +(1-p)\cdot ([M_{i+1}-M_i](0, \ei))^2&\\
& \leq p(1-p)^2(2^{t+1} + O(1/n))^2 + (1-p)p^2(2^{t+1} + O(1/n) )^2&\\
& \leq p(1-p)(4^{t+1} + O(1/n))&\\
& \leq (c/n)(4^{t+1} + O(1/n))
\end{flalign*}

Thus by Lemma \ref{BVMI},

\begin{flalign*}
Pr\{\abs{M_m-M_0}\geq s\cdot n^{\frac{2}{3}}\}
& \leq 2 \exp(-\frac{s^2\cdot n^{\frac{4}{3}}}{(n^2-n) \cdot  (c/n)(4^{t+1} + O(1/n)) + 2/3\cdot (2^{t+1}+O(1/n))\cdot s \cdot n}) &\\
& \leq 2\exp(-\frac{s^2\cdot n^{\frac{4}{3}}}{(n-1)c4^{t+1}  + (2/3)sn^{2/3}2^{t+1}+O(1)}) &\\
& \leq 2\exp(-\frac{s^2\cdot n^{\frac{4}{3}}}{nc4^{t+1}  + (2/3)sn^{2/3}2^{t+1}+O(1)}) &\\
& \leq 2\exp(-\frac{s^2\cdot n^{\frac{4}{3}}}{n(c4^{t+1}  + (2/3)sn^{-1/3}2^{t+1})+O(1)}) &\\
& \leq 2\exp(-\frac{s^2\cdot n^{\frac{1}{3}}}{(c4^{t+1}  + (2/3)sn^{-1/3}2^{t+1})+O(1/n)}) 
\end{flalign*}

\end{proof}

\begin{proof}\textit{Proof of Theorem \ref{maintheorem}} Theorem~\ref{maintheorem} follows from Theorem \ref{SMI} and Lemma \ref{criticallemma} via the following inequalities.
\begin{flalign*}
Pr\{\abs{Y_m-Y_0}\geq s\} 
&=Pr\{\abs{Y_m-Y_0}\geq s \: and \:  \, \neg N^{\prime}\} + Pr\{\abs{Y_m-Y_0} \geq s \: and \: N^{\prime}\} &\\
& \leq Pr\{\abs{M_m-M_0}\geq s\} + Pr\{N^{\prime}\} 
\end{flalign*}
 Thus
\[Pr\{\abs{Y_m-Y_0}\geq s \cdot n^{\frac{2}{3}}\}\leq 2\exp\pth{-\frac{s^2\cdot n^{\frac{1}{3}}}{(c4^{t+1}  + (2/3)sn^{-1/3}2^{t+1})+O(1/n)}}+O(1/n)=o_n(1).\]
\end{proof}

Now set $X_0^\prime := \abs{f_0(R_{t}(X))}-\abs{f_0(R_{t+1}(X))}$.

\begin{lemma}
\label{conc}
For any $s>0$, $Pr\{\abs{X_0^\prime-\E[X_0^\prime]}>sn^{\frac{2}{3}}\}<o_n(1).$
\end{lemma}

\begin{proof}[Proof of Lemma~\ref{conc}]
Observe that
\begin{flalign*}
\abs{X_0^\prime-\E[X_0^\prime]}>t 
&\implies \abs{(\abs{f_0(R_{t-1}(X))}-\abs{f_0(R_{t}(X))})-(\E[\abs{f_0(R_{t-1}(X))}]-\E[\abs{f_0(R_{t}(X))}])}>t\\
&\implies \abs{\abs{f_0(R_t(X)}-\E[\abs{f_0(R_t(X))}]}+\abs{\abs{f_0(R_{t-1}(X)}-\E[\abs{f_0(R_{t-1}(X))}]}>t\\
& \implies  \abs{\abs{f_0(R_t(X)}-\E[\abs{f_0(R_t(X))}]}>t/2 \quad \text{or} \quad   \abs{\abs{f_0(R_{t-1}(X)}-\E[\abs{f_0(R_{t-1}(X))}]}>t/2
\end{flalign*}
Hence, by union bound, \[Pr\{\abs{X_0^\prime-\E[X_0^\prime]}>sn^{\frac{2}{3}} \}\] \[\leq Pr\{\abs{\abs{f_0(R_{t}(X)}-\E[\abs{f_0(R_{t}(X))}]}>(s/2)n^{\frac{2}{3}}\} + Pr\{\abs{\abs{f_0(R_{t-1}(X)}-\E[\abs{f_0(R_{t-1}(X))}]}>(s/2)n^{\frac{2}{3}}\}\leq o_n(1)\]

\end{proof}

Let $X_0$ be the random variable that denotes the number of dominated vertices at the end of the first epoch. Clearly $0 \leq X_0 \leq \abs{f_0(R_{t}(X))}-\abs{f_0(R_{t+1}(X))}=X_0^\prime$. As $X_0^\prime \leq \E[X_0^\prime] + o(n)$ a.a.s. we get that $0 \leq X_0
\leq  \E[\abs{f_0(R_{t}(X))}]-\E[\abs{f_0(R_{t+1}(X))}] + o(n)$ a.a.s.

\section{Second Epoch}
\label{sec:second-epoch}
The second epoch will be a slower version of the first epoch. Here a dominated vertex is chosen uniformly randomly and is removed. The process continues until there is no more dominated vertices.
Similar to the proof of~\cite{DBLP:journals/rsa/AronshtamL16}, our strategy shall be to show that when a dominated vertex is deleted, the expected number of newly created dominated vertices is 
strictly less than $1$, so that within $o(n)$ steps, the strong collapse process comes to a halt. Thus the size of the core will be -- up to a $o(n)$-factor -- the number of vertices remaining after the 
first epoch. 


Let after $t$ pruning phases the first epoch ends and the second epoch begins. Also, Let $Y_i$ be the random variable that denotes number of newly generated dominated vertices solely by the deletion of the dominated vertex at the $i$-th step of the second epoch. Note that $Y_i \in \{0,\cdots,n\}$. 
First we try to calculate $Pr\{Y_i=1\}$.

\begin{lemma}
\label{drift}
For any $i\leq \frac{1}{12}n(1-\gamma)(1-c\gamma)$,we have $\E[Y_i]\leq 1 -\frac{3}{4}(1-c\gamma) + o_n(1)< 1+o_n(1)$.
\end{lemma}

\begin{proof}[Proof of Lemma~\ref{drift}]

We shall say that a vertex is affected by $i$th collapse in the second epoch if its degree is changed by that collapsing step. Define $Q_i$ be the subset of the event $\{Y_i=1\}$ that the newly generated vertex by the $i$th collapse is affected for the first time. Clearly the event $\{\{Y_i=1\} \cap Q_i\}$ represents the fact that only one vertex, say $v$, is newly generated by the $i$th collapse (of vertex $u$) and $v$ is affected for the first time. Thus that particular vertex retains the local structure since the first epoch. The idea here is that the edge $\{u,v\}$ can be attached to any of the possible places after the first epoch ends. We are only calculating the probability that is is attached to a suitable vertex of $ R_t(X)\setminus \{\{u,v\},\{u\}\} $.  To calculate $Pr\{\{Y_i=1\} \cap Q_i\}$ we shall further partition it into two events. To this end, define $P$ be the event that $deg(v)=2$ after $i-1$ steps.  Thus the probability  $Pr\{\{Y_i=1\} \cap Q_i \cap P\}$ is the ratio of the numbers of degree one vertex in $R_t(X)\setminus \{\{u,v\},\{u\}\} $ to the number of non-isolated vertices after $t-1$ phase of the first epoch, as done in eq. 8 of \cite{DBLP:journals/rsa/AronshtamL16}. It can be shown, by using similar arguments like section \ref{sec:concentration}, that both these quantities are concentrated around their mean. These two quantities are, respectively, equal to $c(1-\gammat)\gamma_{t+1}n+o(n)$ and $(1-\gammat)n + o(n)$ a.a.s. 

For the event $\{Y_i=1\} \cap Q_i \cap \overline{P}$ to occur $v$ must be a part of a $2$-simplex. Now we shall bound the number of $2$-simplices remaining after the first epoch. Observe that during the collapsing phases number of $2$-simplices can only decrease. Let us define the random variables $T:=\abs{f_2(X)}$ and $T^\prime:=\abs{f_2(R_t(X))}$ for $X \sim X(n,c/n)$.  Clearly $T^\prime \leq T$. From Markov's inequality we get that $Pr\{T> \log(n)\} \leq O(1/\log(n))$. Thus,  $Pr\{T^\prime> \log(n)\} \leq O(1/\log(n))$. So $Pr\{\{Y_i=1\} \cap Q_i \cap \overline{P}\}\leq O(\log(n))/(1-\gammat)n$ a.a.s. Thus, a.a.s. $Pr\{\{Y_i=1\} \cap Q_i\}\leq  \frac{c(1-\gammat)\gamma_{t+1}n}{(1-\gammat)n} + O(\log(n))/(1-\gammat)n \leq \frac{c(1-\gammat)\gamma_{t+1}}{(1-\gammat)} + o_n(1)$.

Now we need to calculate $Pr\{\{Y_i=1\} \cap \overline{Q_i}\}$. To do this we shall again partition this event into two disjoint events. Let $A_i$ denote the number of affected vertices at $i$th step of the second epoch. Now define the event $B_i:= \bigcap_{j=1}^{i}\{A_j\leq 3\}$. So, $ Pr\{\{Y_i=1\} \cap \overline{Q_i} \cap B_{i-1}\}\leq \frac{3(i-1)}{n(1-\gammat)}$.

To calculate $ Pr\{\{Y_i=1\} \cap \overline{Q_i} \cap \overline{B_{i-1}}\}$ first observe that  $ Pr\{\{Y_i=1\} \cap \overline{Q_i} \cap \overline{B_{i-1}}\} \leq Pr\{ \overline{B_{i-1}}\} \leq \Sigma_{j=1}^{i-1}\{A_j \geq 4\}$ . But for $\{A_i \geq 4 \}$ to happen the corresponding dominated vertex $u$ must be a part of the following arrangement.

\begin{center}
\begin{tikzpicture}[scale=0.4]

\tikzstyle{every node}=[font=\small]

\tikzset{vertex/.style = {shape=circle,draw,inner sep=0 pt, minimum size=15 pt}}
\tikzset{edge/.style = {-> ,> = triangle 60}}
\node[vertex] (a) at  (3,3) {$a$};
\node[vertex] (u) at  (3,0) {$u$};
\node[vertex] (ck) at  (1,-2) {$c_k$};
\node[vertex] (ci) at  (5,-2) {$c_i$};
\node[vertex] (cj) at  (7,-2) {$c_j$};

\draw[edge] (a) to (u);



\tikzset{edge/.style = {- ,> = latex'}}

\draw[edge] (a) to (ck);
\draw[edge] (u) to (ck);

\draw[edge] (a) to (ci);
\draw[edge] (u) to (ci);

\draw[edge] (a) to (cj);
\draw[edge] (u) to (cj);


\end{tikzpicture}
\end{center}

Let $S$ and $S^\prime$ denote the number of such arrangements in $X$ and $R_t(X)$, respectively, for $X \sim X(n,c/n)$. Clearly, $S^\prime \leq S$. From Markov's inequality we get $Pr\{S\geq 1\}\leq O(1/n^2)$. Thus, $Pr\{S^\prime \geq 1\}\leq O(1/n^2)$. Hence, a.a.s. $\{\overline{B_i}\}$ never happens.

Similar argument combined with lemma \ref{geometriclemma1} gives that $Pr\{Y_i\in \{2, \cdots , n\}\} \leq O(1/n^2)$.

By collecting all the terms we have the following inequality.

\begin{flalign*}
\E[Y_i]
&=\Sigma_{j=0}^n j\cdot Pr\{Y_i=j\}\\
&\leq Pr\{Y_i=1\}+ n\cdot O(1/n^2)\\
&\leq Pr\{\{Y_i=1\} \cap Q_i\} + Pr\{\{Y_i=1\} \cap \overline{Q_i}\}+ n\cdot O(1/n^2)\\
&\leq Pr\{\{Y_i=1\} \cap Q_i\} + Pr\{\{Y_i=1\} \cap \overline{Q_i} \cap B_{i-1}\}+Pr\{\{Y_i=1\} \cap \overline{Q_i} \cap \overline{B_{i-1}}\} + n\cdot O(1/n^2)\\
&\leq \frac{c(1-\gammat)\gamma_{t+1}}{1-\gammat} + \frac{3(i-1)}{n(1-\gammat)} + O(1/n) + o_n(1)\\
&\leq c\gamma_{t+1}+  \frac{3(i-1)}{n(1-\gamma)} + o_n(1)\\
&\leq c\gamma+ \epsilon_i + o_n(1)\\
&\leq c\gamma+\frac{1}{4}(1-c\gamma)+o_n(1)\\
&\leq 1-\frac{3}{4}(1-c\gamma) + o_n(1)< 1+o_n(1),  \hspace{3pt}\text{by lemma \ref{drift1}.}
\end{flalign*}
\end{proof}

Note that for any fixed $i$, $\epsilon_i=o_n(1)$ and for $c > 1$ we have $c\gamma<1$.

Let $X_i$ be the number of dominated vertices at the end of the $i$th step of the second epoch. Then we have 

\[X_i=X_{i-1}-1+Y_{i}=X_0-n+\Sigma_{i=1}^nY_i\]

Untill the second epoch ends. If the second epoch stops at $i^\prime$th step then $X_j = 0 \quad \forall j >  i^\prime$.

 Thus we have

\begin{flalign*}
\E[X_0]
&\leq \E[\abs{f_0(R_{t}(X))}]-\E[\abs{f_0(R_{t+1}(X))}]\\
&= (1-\gamma_{t+1}-c\gammat+c\gammat^2)n - (1-\gamma_{t+2}-c\gamma_{t+1}+c\gamma_{t+1}^2)n
\end{flalign*}

and, 

\[\E[X_i]=\E[X_{i-1}]-1+Y_{i}=\E[X_0]-n+\Sigma_{i=1}^nY_i\]
as long as the second epoch continues.

Now let us define $\delta \equiv \delta(t) := (1-\gamma_{t+1}-c\gammat+c\gammat^2) - (1-\gamma_{t+2}-c\gamma_{t+1}+c\gamma_{t+1}^2) $ so that $\E[X_0]\leq n\delta$.

Now we present the main lemmas of this section. The first lemma below shows that with high probability, for any sufficiently small $\e>0$, we can choose a sufficiently large $t$, 
such that the number of vertices deleted in the second epoch is less than $\e n$. The proof is by modelling the number of remaining dominated vertices after $i$ steps of the epoch, 
as a biased random walk. 

\begin{lemma}
\label{endinglemma}
$\forall 0< \epsilon < \min\{\frac{1}{12}(1-\gamma)(1-c\gamma),\frac{5}{192}(1-\gamma)(1-c\gamma)^2\} \quad \exists T$ such that $\forall t>T$ a.a.s. at most $\epsilon n $ vertices will be deleted from $R_t(X)$ before algorithm reaches the core.
\end{lemma}
\begin{proof}[Proof of Lemma~\ref{endinglemma}]
Choose $t$ such that $\delta \equiv \delta(t)<\frac{1}{8}\epsilon(1-c\gamma)$. This can be done because as $t$ increases $\gamma_{t}-\gamma_{t+1}$ and $\gamma_{t+1}-\gamma_{t+2}$ approaches zero.

Now suppose that the second epoch runs for $\rho=\epsilon n$ steps. We shall show that $\E[X_\rho]=0$ a.a.s., i.e., there is no more dominated vertex left to be deleted. To this end we define a sequence of new random variable $\{Z_i\}$ as follows: \[Z_0:=X_0\] and, \[Z_i:=Z_{i-1}-1+Y_{i}=Z_0-n+\Sigma_{i=1}^nY_i\]

Note that $X_i\leq Z_i$ and $Z_i\leq 0 \implies X_i=0$.

As $\rho=\epsilon n \leq \frac{1}{4}n(1-\gamma)(1-c\gamma)$, at the end of the $\rho $ steps number of dominated vertices remaining is

\begin{flalign*}
\E[X_\rho]\leq\E[Z_\rho]
&=\E[Z_0]-\Sigma_i^\rho1+\Sigma_i^\rho \E[Y_i]\\
&\leq n\delta -\frac{3}{4}\rho(1-c\gamma)  &\text{by lemma \ref{drift}}\\
&\leq \frac{1}{8}n\epsilon(1-c\gamma) -\frac{3}{4} n\epsilon(1-c\gamma)\\
&\leq -\frac{5}{96}(1-\gamma)(1-c\gamma)^2n\leq 0
\end{flalign*}
Hence $\E[X_\rho]=0$.

Now we shall proving the concentration. Let us define $l:= \frac{5}{96}(1-\gamma)(1-c\gamma)^2$. Next we proceed to show \[Pr\{Z_\rho>0\}\leq Pr\{Z_\rho-\E[Z_\rho]>l\cdot n\}<o_n(1)\]

Write $Z_\rho$ as $Z_\rho \equiv Z_0 + Z_\rho^\prime(Y_1,\cdots,Y_\rho)$.
First observe that, from lemma \ref{conc}, $Pr\{Z_0-\E[Z_0]>(l/2)n\} \leq Pr\{\abs{Z_0-\E[Z_0]}>(l/2)n\}\leq o_n(1)$ for any $s>0$.

From the main geometric lemma we get
\begin{flalign*}
\E[e^{t(Y_i-\E[Y_i])}]\leq \E[e^{tY_i}]
&\leq 1 + e^t + O(\frac{e^{2t}}{n^2}\cdot \frac{1-(e^t/n^3)^{n-1}}{1-e^t/n^3})
\end{flalign*}

Fix $t>0$. Then, 

\begin{flalign*}
Pr\{Z_\rho^\prime-\E[Z_\rho^\prime]>(l/2)\cdot n\}
&=Pr\{\Sigma_i^\rho Y_i - \Sigma_i^\rho \E[Y_i]>(l/2)\cdot n\}\\
&=Pr\{e^{t(\Sigma_i^\rho Y_i - \Sigma_i^\rho \E[Y_i])}>e^{tn(l/2)}\}\\
&\leq e^{t(\Sigma_i^\rho Y_i - \Sigma_i^\rho \E[Y_i])}/e^{tn(l/2)}\\
&\leq \frac{(1 + e^t + O(\frac{e^{2t}}{n^2}\cdot \frac{1-(e^t/n^3)^{n-1}}{1-e^t/n^3}))^\rho}{e^{tn(l/2)}}
\end{flalign*}

setting $t=\log(n)$ we get
\begin{flalign*}
Pr\{Z_\rho^\prime-\E[Z_\rho^\prime]>(l/2)\cdot n\}
&\leq \frac{(1 + n + O(\frac{1-(1/n^2)^{n-1}}{1-1/n^2}))^\rho}{n^{n(l/2)}}\\
&\leq \frac{(1 + n + O(\frac{1}{1-1/n^2}))^{\epsilon n}}{n^{n(l/2)}}\\
&\leq \frac{((O(1) + n)^\epsilon)^n}{(n^{(l/2)})^n}\\
&\leq (\frac{(O(1) + n)^\epsilon}{n^{(l/2)}})^n\leq o_n(1)
\end{flalign*}

As $\epsilon < l/2$ ,the quantity approaches zero as $n$ increases, thus the last inequality follows. So, \[Pr\{Z_\rho-\E[Z_\rho]>l\cdot n\}\leq Pr\{Z_0-\E[Z_0]>(l/2)\cdot n\} + Pr\{Z_\rho^\prime-\E[Z_\rho^\prime]>(l/2)\cdot n\}\leq o_n(1)\]

\end{proof}

The next lemma follows from properties of $\gamma_t$ and the concentration bounds presented in Section~\ref{sec:concentration}.

\begin{lemma}\label{approximationlemma}
$\forall 0< \delta  \quad \exists T$ such that $\forall t>T$ a.a.s.  \[(1-\gamma)(1-c\gamma) n + o(n) \leq \abs{f_0(R_{t}(X))} \leq (1-\gamma)(1-c\gamma)n + \delta n + o(n). \]
\end{lemma}
\begin{proof}[Proof of Lemma~\ref{approximationlemma}]

From theorem \ref{expectation} and theorem \ref{maintheorem} it can the shown that a.a.s. 
\[\abs{f_0(R_t(X))}= (1-\gamma_{t+1}-c\gammat+c\gammat^2)n + o(n) \]

Define $h(x):=1-e^{-c(1-x)}-cx(1-x)$ on the interval $(0,1)$. It can be checked that on this interval $h^\prime(x)<0$, i.e., $h(x)$ is strictly decreasing. Thus we have that $(1-\gamma_{t+2}-c\gamma_{t+1}+c\gamma_{t+1}^2) \leq (1-\gamma_{t+1}-c\gammat+c\gammat^2) $. Hence the left inequality follows. 

Note that $\{\gammat\}_t$ is a monotonically increasing sequence that converges to $\gamma$. Therefore, $\{ (1-\gamma_{t+1}-c\gammat+c\gammat^2)\}_t$ is a monotonically decreasing sequence that converges to $(1-\gamma-c\gamma+c\gamma^2)$. Thus, by choosing sufficiently large $T$, we can restrict  $\{ (1-\gamma_{t+1}-c\gammat+c\gammat^2)\}_{t>T}$ inside a $\delta$-ball round $(1-\gamma-c\gamma+c\gamma^2)$ for any $\delta>0$. Hence the right inequality follows.
\end{proof}

Using above two lemmas we have the proof of our first main result Theorem~\ref{sizeofcore} about the size of the core (after strong collapse) of a ER complex.
\begin{proof}[Proof of Theorem \ref{sizeofcore}]

By Lemma \ref{endinglemma} and Lemma \ref{approximationlemma}. 
\end{proof}

\section{End Range Phase Transition}
\label{sec:end-range-phase-trans}
Let $P(X)$ denote the number all possible dominated-dominating pairs. It can be shown that for $X \sim X(n,p)$, $\E[P(X)]=(n(n-1)/2)p(1-p(1-p))^{n-2}$. For $p=\frac{\lambda \log(n)}{n}$, where $\lambda > 1$, $\E[P(X)]=O(\frac{n\log(n)}{n^\lambda})=o_n(1)$. Thus, by Markov's inequality, a.a.s. there is no dominated vertex to start the collapsing procedure.

  Now we shall focus on the behavior of $X \sim X(n,p)$ where $p=1-\frac{\lambda \log(n)}{n}$. For $\lambda > 2$, $\E[P(X)]=O(\frac{n^2}{n^\lambda})=o_n(1)$. Thus a.a.s. there is no dominated vertex to start the collapsing procedure. But the situation is quite opposite when $\lambda < 1$ as the following lemma claims. 

\begin{lemma}
\label{l:end-ph-trans}
For $X \sim X(X, 1-\frac{\lambda \log n}{n})$ and $\lambda < 1$, a.a.s. $X$ is collapsible.
\end{lemma}

\begin{proof}[Proof of Lemma~\ref{l:end-ph-trans}]
We shall show that, in this range, a.a.s. there exits a vertex adjacent to every other vertices. Let us define the random variable $V \in \{0, \cdots, n\}$ that counts the number vertices that are adjacent to all other vertices. Clearly $\E[V]=np^{n-1}=n(1-\frac{\lambda \log(n)}{n})^{n-1}=\Theta(\frac{n}{n^\lambda})$. Now we shall calculate $Var(V)$. Let $I_i$ denote the indicator random variable that $v_i$ is adjacent to all other vertices. Then, 

\begin{flalign*}
Var(V)
&=nVar(I_1)+n(n-1)cov(I_1,I_2)\\
&=np^{n-1}(1-p^{n-1}) + n(n-1)(p^{2n-3}-p^{2n-2})
\end{flalign*}

Thus, 
\begin{flalign*}
Pr\{V=0\}
&\leq \frac{Var(V)}{(\E[V])^2}\\
&\leq \frac{np^{n-1}(1-p^{n-1}) + n(n-1)(p^{2n-3}-p^{2n-2})}{n^2p^{2n-2}}\\
&\leq 1/(np^{n-1})+(1/p-1)\\
&\leq \Theta (\frac{n^\lambda}{n}) + \frac{\lambda \log(n)}{n- \lambda \log(n)}=o_n(1)
\end{flalign*}

Thus $V \geq 1$ a.a.s.

\end{proof}

	
\bibliographystyle{plainurl}	
\bibliography{mybib}

\end{document}